\documentclass[letterpaper,10 pt,conference]{ieeeconf}  
\IEEEoverridecommandlockouts                              %

\usepackage{cite}
\usepackage{amsmath,amssymb,amsfonts}
\usepackage{algorithm}
\usepackage{algorithmic}

\usepackage{graphicx}
\usepackage{textcomp}
\usepackage{xcolor}
\def\BibTeX{{\rm B\kern-.05em{\sc i\kern-.025em b}\kern-.08em
    T\kern-.1667em\lower.7ex\hbox{E}\kern-.125emX}}

\usepackage{comment}

\usepackage{subfigure}

\usepackage{xcolor}
\usepackage{xparse}
\usepackage{mathrsfs}
\usepackage{graphics} 
\usepackage{amsmath} 
\usepackage{amssymb}  
\usepackage{amsthm}
\usepackage{enumitem} 
\usepackage{hyperref}%
\hypersetup{colorlinks=true,  linkcolor=blue,  breaklinks=true,  urlcolor=blue,  citecolor=blue}

\usepackage{empheq} 

\usepackage{soul}
\usepackage{url}

\usepackage[prependcaption, colorinlistoftodos]{todonotes}

\newcommand\oprocendsymbol{\hbox{$\triangle$}}
\newcommand\oprocend{\relax\ifmmode\else\unskip\hfill\fi\oprocendsymbol}

\DeclareSymbolFont{bbold}{U}{bbold}{m}{n}
\DeclareSymbolFontAlphabet{\mathbbold}{bbold}
\newcommand{\vect}[1]{\mathbbold{#1}}

\newtheorem{theorem}{Theorem}
\newtheorem{remark}[theorem]{Remark}
\newtheorem{example}{Example}

\newtheorem{definition}[theorem]{Definition}
\newtheorem{proposition}[theorem]{Proposition}
\newtheorem{corollary}[theorem]{Corollary}
\newtheorem{problem}{Problem}

\newcommand {\be}{\begin{equation}}
\newcommand {\ee}{\end{equation}}

\newcommand{\x}{{\bf x}}
\newcommand{\ub}{{\bf u}}
\newcommand{\y}{{\bf y}}

\newcommand{\myalpha}{\boldsymbol{\alpha}}
\newcommand{\mybeta}{\boldsymbol{\beta}}

\title{\LARGE\bf   Reference Output Tracking in Boolean Control Networks}

\author{Giorgia Disar\`o and Maria Elena Valcher
\thanks{G. Disar\`o and  M.E. Valcher are with the Dipartimento di Ingegneria dell'Informazione,
 Universit\`a di Padova,
    via Gradenigo 6B, 35131 Padova, Italy, e-mail:  \texttt{disarogior@dei.unipd.it, meme@dei.unipd.it}}
   }
  \date{}

\begin{document}
\maketitle

\begin{abstract}
In this paper, the problem of tracking a given reference output trajectory is investigated for the class of Boolean control networks, by resorting to their algebraic representation. First, the case of a finite-length reference trajectory is addressed, and the analysis and algorithm first proposed in \cite{FT_track_BCN} are extended to be able to deal with arbitrary initial conditions and to identify all possible solutions.
The approach developed for the finite-length case is then adjusted to cope with periodic reference output trajectories.
The results of the paper are illustrated through an example.
\end{abstract}


 \section{Introduction} \label{intro}

A central and longstanding problem in control theory is the output tracking problem - that is, the design of a control law that ensures that the system output follows a prescribed reference trajectory. This challenge remains equally significant in the context of Boolean control networks (BCNs) \cite{BCNCheng}, a class of systems governed by discrete logical dynamics. In BCNs, the system states, control inputs, and outputs assume Boolean values, and the network evolves according to logical update functions rather than differential or difference equations. In this setting, output tracking consists in designing a control sequence that drives the system output to coincide with a desired Boolean reference trajectory, either instantaneously or after a finite transient period.

Various approaches have been proposed to characterize the conditions under which exact or approximate output tracking can be achieved in BCNs. In particular, the case of a constant reference output signal - often referred to as output regulation - has been extensively investigated  by resorting to the algebraic representation of a BCN \cite{Cheng2009,lin-rep-dyn-Boolean} that lays its foundations on the semi-tensor product (STP) introduced by Cheng in \cite{Cheng2009}. For example, in \cite{ACC2014_BCN} and \cite{Li2015}, the authors reformulated the output regulation problem as a set-stabilization problem, allowing the design of state feedback controllers that drive the network output to a fixed desired value.
The problem of tracking a constant reference signal was also investigated in \cite{PanEtAl2023}, where the authors proposed a framework capable of handling situations in which perfect tracking is infeasible. Specifically, the problem is formulated as an optimal tracking problem, and a suitable control scheme is designed to ensure that each state evolves toward its optimal output tracking trajectory, namely the trajectory obtained by stabilizing the BCN to the cycle that minimizes the number of output mismatches.
However, when moving to the more general and challenging scenario where the reference output trajectory is time-varying, the output tracking problem can no longer be reduced to set-stabilization, and a different analytical and algorithmic treatment is required.

The output tracking problem of finite-length trajectories in BCNs was first investigated in \cite{FT_track_BCN}. In this paper, the authors provide necessary and sufficient conditions for the trackability of a finite-length (but otherwise arbitrary) reference output trajectory from a given initial state, along with an algorithmic procedure to construct a control sequence that achieves the desired tracking objective, whenever possible. If the reference output trajectory is not trackable, the output tracking problem is formulated as an optimization problem, and an optimal control sequence that minimizes the tracking error is designed. In \cite{Liu23}, the exact output tracking results from \cite{FT_track_BCN} were extended to the case where the control input influences not only the state dynamics but also the output dynamics. 

A first extension to the case of an infinite time-varying output trajectory was presented in \cite{Li17}, where the reference signal is generated by an external Boolean network, that is an autonomous logical system. The authors derived necessary and sufficient conditions for the solvability of the problem - based on a matrix equation that is, however, not easily solvable - and proposed a control design method relying on an augmented system formulation. An alternative approach to the same problem was proposed in \cite{track_Cheng}. The approach in \cite{track_Cheng} first computes the set of global attractors and the maximal control invariant subset within a preassigned region. Based on these computations, a necessary and sufficient condition is then derived using the set controllability method applied to an auxiliary system that integrates the original BCN with the reference network. 

In this paper, we revisit the solution proposed in \cite{FT_track_BCN} for finite-length reference output trajectories with the aim of making it suitable for addressing a broader class of output tracking problems. In \cite{FT_track_BCN}, the analysis is restricted to the case where the initial condition is fixed a priori and exact tracking is required, meaning that no delay is allowed. When the reference output trajectory is trackable from the specified initial condition, a specific control sequence achieving the desired tracking objective is constructed. To this end, we propose a modified version of the algorithm in \cite{FT_track_BCN} that explicitly accounts for arbitrary initial conditions by simultaneously addressing all possible initial states and systematically identifying {\em all} control sequences that achieve exact tracking. Moreover, the proposed formulation has the advantage of being more readily extendable to scenarios in which exact tracking is not feasible but delayed tracking remains achievable, which may still be of practical relevance in many applications. Finally, the modified algorithmic framework naturally extends, with limited effort, to the more challenging and interesting problem of periodic output trajectory tracking.   

The remainder of this paper is organized as follows. In Section \ref{prel}, we give some preliminary notions about Boolean control networks. Section \ref{sec:track_mb} addresses output tracking for finite-length trajectories, whereas Section \ref{sec:periodic} is devoted to the problem of periodic output trajectory tracking. Finally, Section \ref{concl} concludes the paper. 
\medskip


{\bf Notation}.
Given  $k$ and $n$ in ${\mathbb Z},$ with $k \le n$, we denote by $[k,n]$ the set of integers $\{k, k+1, \dots, n\}$, and we use $k\bmod n$ to denote the remainder of the Euclidean division of $k$ by $n$. 
We focus on {\em Boolean vectors and matrices}, whose elements belong to the set $\mathcal B \triangleq \{0,1\}$, and adopt the standard logical operations: sum (OR)  $\vee$, product (AND) $\land$, and negation (NOT) $\neg$. The symbol $\delta_k^i$ represents the $i$-th {\em canonical vector} in $\mathbb{R}^k$, and we let $\mathcal L_k$ denote the set of all $k$-dimensional canonical vectors. Similarly, $\mathcal L_{k \times q} \subset \mathcal B^{k \times q}$ is the set of all $k \times q$ {\em logical matrices}, meaning matrices whose columns are canonical vectors of length $k$. Any matrix $L \in \mathcal L_{k \times q}$ can be written as $L = [\delta_k^{i_1} \ \delta_k^{i_2} \ \dots \ \delta_k^{i_q}]$, with  $i_1,i_2,\dots, i_q\in [1,k]$.   The $n\times n$ identity matrix is denoted by $I_n$. The vector of dimension $k$ whose entries are all ones is written $\vect{1}_k$, while the vector of zeros is $\vect{0}_k$. For a matrix $M$, its $(i,j)$-th element is $[M]_{ij}$, and for a vector $\mathbf{v}$, its $j$-th entry is $[\mathbf{v}]_j$. A sequence of vectors $\mathbf{v}(i), \mathbf{v}(i+1), \dots, \mathbf{v}(j)$, with $i,j \in \mathbb{Z}$ and $i \le j$, is denoted by $\{{\bf v}(t)\}_{t=i}^{j}$. Given two vectors ${\bf v}$ and ${\bf w}$ in $\mathcal B^n$, 
we write ${\bf v} > {\bf w}$  to indicate that each component of ${\bf v}$ is greater than the corresponding component of ${\bf w}$, i.e., $[{\bf v}]_j > [{\bf w}]_j, \forall j \in [1,n]$. 

There exists a bijective mapping between Boolean variables $X \in \mathcal B$ and  vectors $\mathbf{x} \in \mathcal L_2$, defined as $\x = [X \ \neg X]^\top$. For two matrices $A \in \mathbb{R}^{m \times n}$ and $B \in \mathbb{R}^{p \times q}$, we define their {\em (left) semi-tensor product} \cite{STP2001} as
$$
A \ltimes B \triangleq (A \otimes I_{l/n})(B\otimes I_{l/p}), \quad l \triangleq {\rm l.c.m.}\{n,p\},
$$
where $\otimes$ is the Kronecker product. 
 The semi-tensor product allows to extend  the  bijection between $\mathcal B$ and $\mathcal L_2$ to a bijection between $\mathcal B^n$ and $\mathcal L_{2^n}$. Specifically, any Boolean vector $X = [X_1, X_2, \dots, X_n]^\top \in \mathcal B^n$ is mapped to
$$\x \triangleq \begin{bmatrix} X_1 \\ \neg X_1\end{bmatrix} \ltimes \begin{bmatrix} X_2 \\ \neg X_2\end{bmatrix} \ltimes \dots \ltimes \begin{bmatrix} X_n \\ \neg X_n\end{bmatrix}.$$
A comprehensive discussion of semi-tensor product properties can be found in \cite{BCNCheng}. The {\em Hadamard (or entry-wise) product} of $A \in \mathbb{R}^{m \times n}$ and $B \in \mathbb{R}^{p \times q}$ is denoted by $A\odot B$.

\section{Boolean control networks: Preliminaries} \label{prel} 

 A {\em Boolean control network (BCN)} is a discrete-time logical system described by the following state-space model:
\begin{subequations} \label{bacon}
\begin{eqnarray}
X(t+1) &=& f(X(t), U(t)),   \label{bcn}\\ 
Y(t) &=& h(X(t)), \qquad \qquad t \in \mathbb Z_+, \label{bcno}
\end{eqnarray}
\end{subequations}
where $X(t)\in \mathcal B^n$ is the $n$-dimensional state variable, $U(t)\in \mathcal B^m$ is the $m$-dimensional input, and 
$Y(t)\in \mathcal B^p$ is the $p$-dimensional output  at time $t$. The mappings $f$ and $h$ are logical functions, with $f: \mathcal B^n \times \mathcal B^m \to \mathcal B^n$ and $h: \mathcal B^n \to \mathcal B^p$.

When no external input is applied, the system becomes autonomous, and the BCN reduces to a {\em Boolean network (BN)}. In this case, the dynamics simplifies to
\be \label{bn} 
X(t+1) = f(X(t)), \quad t \in \mathbb Z_+, 
\ee
where 
$f: \mathcal B^n \to \mathcal B^n$ remains a  logical function. 
\\
By exploiting the bijective correspondence between Boolean  and logical vectors, 
the BN \eqref{bn} 
 can be expressed in its {\em algebraic form} \cite{BCNCheng} as 
 \be
 \label{bna} 
 \x (t+1) = L \x (t), \quad t\in \mathbb Z_+,
\ee
where $\x(t) \in \mathcal L_{N}$ and $L\in {\mathcal L}_{N\times N}$,  $N \triangleq 2^n$.
Analogously, the algebraic representation of the BCN in \eqref{bacon} is given by 
 \begin{subequations}\label{BCNtot}
\begin{eqnarray} 
 \label{bcnA} 
\x (t+1) &=& L \ltimes \ub(t) \ltimes \x(t), 
 \quad t\in \mathbb Z_+,\\
 \y(t) &=& H\x(t), \label{bcnAo}
 \end{eqnarray}
 \end{subequations}
where $\x(t) \in \mathcal L_{N}$, $\ub(t) \in \mathcal L_{M}$, $\y(t) \in \mathcal L_{P}$, $L\in \mathcal L_{N\times NM}$ and $H\in {\mathcal L}_{P\times N}$, with $N \triangleq 2^n$, $M \triangleq 2^m$ and $P\triangleq 2^p$.

The logical matrix $L\in \mathcal L_{N\times NM}$ can be decomposed into $M$ square submatrices of dimension $N$ as 
\be \label{Lblock}
L = \left[\begin{array}{c|c|c|c}
L_1 & L_2 & \dots & L_M
\end{array}\right],
\ee
where  $L_i \in \mathcal L_{N\times N}, i \in [1,M]$,  is the transition matrix of the $i$-th subsystem of the BCN, i.e., the Boolean network corresponding to the constant input $\ub(t) = \delta_M^i, \forall t \in \mathbb Z_+$:
\be \label{Li}
\x (t+1) = L_i \x (t), \quad t\in \mathbb Z_+.
\ee
We associate with the BCN \eqref{BCNtot}
 the Boolean (not necessarily logical) matrix
\begin{equation} \label{ltot}
L_{\text{tot}} \triangleq L_1 \vee L_2 \vee \dots \vee L_M \in {\mathcal B}^{N \times N},
\end{equation}
which captures all possible state transitions that can occur in the network, corresponding to some input value.  Specifically, $[L_{tot}]_{\ell j} =1$ if and only if
$\delta^\ell_N$ is the {\em successor} of $\delta^j_N$ corresponding to some input value, i.e., there exists $\ub=\delta^i_M\in {\mathcal L}_M$ such that $L\ltimes \delta^i_M \ltimes \delta^j_N = \delta^\ell_N$. The successors of a state $\delta_N^j$ correspond to the indices of the nonzero entries of the vector $L_{tot} \delta_N^j \in \mathcal B^N$.    On the other hand, a state $\delta_N^\ell$ is a {\em predecessor} of a state $\delta_N^j$ if
there exists $\ub \in {\mathcal L}_M$ such that $L\ltimes \ub \ltimes \delta^\ell_N = \delta^j_N$, and this is the case if
and only if $[L_{tot}^\top \delta_N^j]_{\ell} = 1$. Consequently, the predecessors of a state $\delta_N^j$ correspond to the indices of the nonzero entries of the vector $L_{tot}^\top \delta_N^j \in \mathcal B^N$.   
\smallskip

A state $\delta_N^{j}$  of the BCN \eqref{BCNtot} is said to be {\em reachable in $k$ steps} \cite{BCNCheng,ChengContr,LASCHOV2012,EFMEV_JCD} from the state $\delta_N^{\ell}$ if there exists an input sequence $\{\ub(t)\}_{t=0}^{k-1}$  
that leads from $\x(0) = \delta_N^{\ell}$ to $\x(k) = \delta_N^{j}$, that is, $\delta_N^{j} = L \ltimes \ub(k-1) \ltimes L \ltimes \ub(k-2) \ltimes \dots \ltimes L \ltimes \ub(0) \ltimes\delta_N^{\ell}$.  The state $\delta_N^{j}$ is reachable in $k$ steps from $\delta_N^{\ell}$ if and only if (see, e.g., Section 3 in \cite{EF_MEV_BCN_Aut} and Chapter 16 in \cite{BCNCheng}) $[L_{tot}^{k}]_{j \ell} > 0 $. 
A set of states ${\mathcal X} = \{\delta_N^{j_1}, \delta_N^{j_2}, \dots, \delta_N^{j_d}\}$ 
is  {\em reachable in $k$ steps} from the state $\delta_N^\ell$ if 
at least one of the states of ${\mathcal X}$ is reachable in $k$ steps from $\delta^\ell_N$.
This is the case if and only if
$(\sum_{h=1}^d \delta^{j_h}_N)^\top  L_{tot}^k \delta^\ell_N$ is positive.

\smallskip

\section{Output Tracking of Finite-
Length Trajectories} \label{sec:track_mb}

In this section, we investigate 
the finite-horizon output tracking problem, first studied in \cite{FT_track_BCN}, by extending the analysis to arbitrary initial conditions, without the need to solve the problem separately for each initial state. To this end, we introduce a modified version of the algorithm in \cite{FT_track_BCN}, which enables us to handle arbitrary initial conditions and determine all possible solutions to the problem. Moreover, it is suitable for further extensions of the output tracking problem, as discussed later (see Remark \ref{justify}).

As a first step, we introduce the notation $\y(k,\x(0),\{\ub(t)\}_{t=0}^{k-1})$ $\triangleq$ $HL \ltimes \ub(k-1)\ltimes L \ltimes \ub(k-2) \ltimes$ $\dots$ $\ltimes L \ltimes \ub(0)\ltimes \x(0)$ to indicate the output of the BCN \eqref{BCNtot} at time $k\in \mathbb Z_+, k\ge 1,$ obtained starting from the initial condition $\x(0)$ and applying the input sequence $\{\ub(t)\}_{t=0}^{k-1}$. We give the definition of trackability of an output trajectory (see Definition 2 in \cite{FT_track_BCN}). 

\begin{definition} \label{trackability_def}
Given a BCN \eqref{BCNtot}, an initial state $\x(0)=\x_0\in {\mathcal L}_N$ and a  finite-length reference output trajectory $\{\y_r(t)\}_{t=1}^{T}$, with $\y_r(t)\in {\mathcal L}_P$,
we say that the reference output trajectory is {\em trackable} from $\x_0$ if there exists a control sequence $\{\ub(t)\}_{t=0}^{T-1}$, with $\ub(t) \in \mathcal L_M$, such that $\y(t,\x_0,\{\ub(\tau)\}_{\tau=0}^{t-1}) = \y_r(t), \forall t \in [1, T]$. 
\end{definition}

\begin{problem} \label{probl1}
Given a finite-length reference output trajectory $\{\y_r(t)\}_{t=1}^{T}$, with $\y_r(t)\in {\mathcal L}_P$, under what conditions is it trackable starting from any initial condition?
\end{problem}
\smallskip

In order to determine necessary and sufficient conditions for the solvability of Problem \ref{probl1}, we first introduce two important definitions.\\
$\bullet$ \ A state trajectory $\{\x_t\}_{t=1}^{T}$ generated by the BCN \eqref{BCNtot} is said to be {\em compatible with the reference output trajectory} $\{\y_r(t)\}_{t=1}^T$ if $H \x_t= \y_r(t)$ for every $t\in [1,T]$.\\
$\bullet$ \ We let
 $\mathcal X_1$ denote the set
of initial states $\x_1$ of all state trajectories $\{\x_t\}_{t=1}^{T}$ generated by the BCN \eqref{BCNtot} and compatible with the reference output trajectory $\{\y_r(t)\}_{t=1}^T$.\\
Condition $\mathcal X_1 \ne \emptyset$  is equivalent to saying that there exists at least one state trajectory of length $T$ that generates the output sequence $\{\y_r(t)\}_{t=1}^T$. 
However, this is not enough to ensure that the tracking problem can be solved starting from any initial state. We must also verify that $\mathcal X_1$ is reachable in one step from every state in $\mathcal L_N$, that is, that  $\forall \x_0 \in \mathcal L_N$ there exists $\ub \in \mathcal L_M$ such that $L \ltimes \ub \ltimes \x_0 \in \mathcal X_1$. This is formalized in the following theorem. 

\begin{theorem} \label{cns_exact}
Given a BCN \eqref{BCNtot} and a finite-length reference output trajectory $\{\y_r(t)\}_{t=1}^{T}$, with $\y_r(t)\in {\mathcal L}_P,$
the following facts are equivalent. 
\begin{itemize}
    \item[i)] The output trajectory $\{\y_r(t)\}_{t=1}^{T}$ is trackable from every $\x(0) = \x_0 \in \mathcal L_N$. 
    \item[ii)] The set $\mathcal X_1$ is non-empty and for every $\x(0)=\x_0 \in \mathcal L_N$ there exists $\ub \in \mathcal L_M$ such that $L \ltimes \ub \ltimes \x_0 \in \mathcal X_1$.
\end{itemize}
\end{theorem}

\begin{proof}
The necessity follows directly from the fact that, in order for a given reference output trajectory to be trackable from every initial state  - i.e., for condition {\em i)} to be satisfied - it is required that (a) the trajectory can actually be generated by the considered BCN \eqref{BCNtot}, meaning that there exists at least one compatible state trajectory, and (b) the set of initial states of such state trajectories can be reached in a single step from any state in the network. The combination of conditions (a) and (b) corresponds precisely to condition {\em ii)} in the theorem.\\
The reverse implication follows similar arguments.
\end{proof}

Once the previous conditions for the problem solvability have been identified, the next goal is to understand how it is possible to actually check them. 
In \cite{FT_track_BCN}, 
an algorithm was proposed to check the trackability of a finite-length reference output trajectory from a specific initial condition. When the test has a positive outcome, the algorithm additionally provides a feasible input trajectory that guarantees output tracking from that specific initial state.

We want to review Algorithm 1 in
\cite{FT_track_BCN} to be able to both test the solvability of the problem {\em for every initial state}, and record {\em all possible state/input trajectories} compatible\footnote{A state/input trajectory $\{ (\x_t,\ub_t)\}_{t=1}^T$ is compatible with the reference output sequence $\{\y_r(t)\}_{t=1}^{T}$ if, for every $t\in [1, T]$, $H\x_t = \y_r(t)$ and, for every $t\in [1, T-1]$, $\x_{t+1}= L\ltimes \ub_t\ltimes \x_t$.} with the given finite-length output sequence $\{\y_r(t)\}_{t=1}^{T}$.

To this end, we introduce, for every $i\in [1,P]$,
the set of states that generate the output value
  ${\bf y}=\delta_P^{i}$, i.e., 
  $${\mathcal C}_i \triangleq \{ \delta^j_N \in {\mathcal L}_N: H \delta^j_N =\delta^i_P\},$$
 which is called the {\em one-step indistinguishability class} of the output ${\bf y}=\delta_P^{i}$ \cite{EF_MEV_BCN_obs2012}.
  As remarked in 
 \cite{FaultDetTAC15},
 the 
set of states 
in ${\mathcal C}_i$
can be represented in vectorized form as 
\be \label{ind_class}
{\bf v}_i \triangleq H^\top \delta_P^{i},
\ee
by this meaning that ${\bf x} =\delta^j_N$ satisfies $H \delta^j_N = \delta^i_P$ if and only if $[{\bf v}_i]_j =1.$ If there exists an output value, say $\delta_P^{\bar \imath}$, that cannot be generated by the BCN  \eqref{BCNtot}, and hence the corresponding one-step indistinguishability class is empty, then the associated vector ${\bf v}_{\bar \imath}$ is the zero vector.  
\\
Now, assume that the reference output trajectory is given by $$\y_r(1) = \delta_P^{i_1},\ \y_r(2) = \delta_P^{i_2}, \ \dots, \ \y_r(T) = \delta_P^{i_T}.$$ We associate with this trajectory the  vector sequence  $\{{\bf v}(t)\}_{t=1}^T \triangleq \{{\bf v}_{i_t}\}_{t=1}^{T}$, with ${\bf v}_{i_t}$ as in \eqref{ind_class}, representing the indistinguishability classes at each time $t \in [1, T]$.
We introduce the vector sequence  $\{{\boldsymbol{\alpha}}(t)\}_{t=1}^{T}$ computed as follows \cite{FT_track_BCN}
\be \label{a_seq}
\begin{cases} 
{\boldsymbol{\alpha}}(1) = {\bf v}(1), \\
{\boldsymbol{\alpha}}(t) = {\bf v}(t) \odot (L_{tot}\ {\boldsymbol{\alpha}}(t-1)), \ t\in [2,T].  
\end{cases}
\ee
It is straightforward to observe, by extending the reasoning in \cite{FT_track_BCN}, that ${\boldsymbol{\alpha}}(1)$, which coincides with ${\bf v}(1)$, contains nonzero  entries only in positions corresponding to the states that belong to the one-step indistinguishability class ${\mathcal C}_{i_1}$ associated with $\y_r(1)$.
At the subsequent time instants $t = 2, \dots, T$, the vector ${\boldsymbol{\alpha}}(t)$ has nonzero entries in correspondence with the states $\delta_N^j$ that belong to the class 
${
\mathcal C}_{i_t}$
and 
that are successors of some state whose index is active in ${\boldsymbol{\alpha}}(t-1)$,
(which means that
$\exists i\in [1,N]$ s.t. $[L_{tot}]_{ji} \ne 0$  and   $[{\boldsymbol{\alpha}}(t-1)]_i \ne 0$).
 Therefore, if ${\boldsymbol{\alpha}}(T) \ne \vect{0}_N$, it means that there exist at least a state $\x_1$ and an input sequence $\{\ub_t\}_{t=1}^{T-1}$ such that starting from $\x(1) = \x_1$ and applying at each time step $t\in [1,  T-1]$ the   input $\ub(t)= \ub_t$, we obtain a state trajectory, say $\{\x_t\}_{t=1}^{T}$, whose associated output trajectory $\{\y_t\}_{t=1}^T$ satisfies $\y_t = \y_r(t), \forall t\in[1,T]$.  

 \begin{remark} \label{cfr_loro}
It is worth emphasizing that, although the vector sequences $\{a(t)\}_{t=0}^{T}$ in \cite{FT_track_BCN} and $\{\boldsymbol{\alpha}(t)\}_{t=1}^{T}$ in this paper are generated according to the same algorithm, they are initialized in different ways and serve two distinct purposes.  Specifically, in \cite{FT_track_BCN}, the authors address the output tracking problem for a {\em given} initial condition, say $\x_0$, and therefore initialize their vector sequence as $a(0) = \x_0$. In that context, the condition $a(T) \ne \vect{0}_N$ is both necessary and sufficient for the solvability of the output tracking problem starting from that particular initial condition $\x_0$. Instead, in the framework considered in this paper, the objective is to track the reference output trajectory starting from {\em any} initial condition. In our case, condition ${\boldsymbol{\alpha}}(T) \ne \vect{0}_N$ remains necessary, but is no longer sufficient. This is because the sequence $\{{\boldsymbol{\alpha}}(t)\}_{t=1}^{T}$ is constructed from time $t = 1$, with initialization ${\boldsymbol{\alpha}}(1) = {\bf v}(1)$. In this case, the nonzero entries of the vector ${\boldsymbol{\alpha}}(T)$ simply indicate the indices of the final states of the state trajectories compatible with the given reference output trajectory stemming from one of the states in $\mathcal X_1$. Consequently, ${\boldsymbol{\alpha}}(T)\ne \vect{0}_N$ ensures only that the set $\mathcal X_1$ of all possible initial states $\x_1$  is not the empty set, or, equivalently, that the reference output trajectory is compatible with the BCN. We need also to check that  $\mathcal X_1$
can be reached in a single step from every state in $\mathcal{L}_N$. This
  requires to preliminarily identify the set ${\mathcal X}_1$.
\end{remark}

  To determine the set $\mathcal X_1$, 
  we cannot use the vector ${\boldsymbol{\alpha}}(1)$ since, 
  in general, it is not true that for every $j\in [1,N]$ such that $[{\boldsymbol{\alpha}}(1)]_j \ne 0$ there exists an input sequence such that the state trajectory stemming from $\x(1) = \delta^j_N$ and corresponding to that input sequence guarantees  output tracking in $[1,T]$.
  So, we need
  to introduce a new vector sequence $\{{\boldsymbol{\beta}}(t)\}_{t=1}^{T}$, derived from $\{{\boldsymbol{\alpha}}(t)\}_{t=1}^{T}$ and computed backward with the goal of replacing with zeros every entry $[{\boldsymbol{\alpha}}(t)]_j$ of ${\boldsymbol{\alpha}}(t)$
that corresponds to a state $\delta^j_N$ for which no successor exists generating $\y_r(t+1)=\delta^{i_{t+1}}_P$ as output. 
This is implemented by means of Algorithm \ref{alg1}, below. 

\begin{algorithm}[h]
\caption{Solvability of Problem \ref{probl1}} \label{alg1} 
\smallskip
\textbf{Input:}  - A BCN as in \eqref{BCNtot}, i.e., a pair $(L,H)$;\\
\hspace*{2.8em} - The sequence $\{{\boldsymbol{\alpha}}(t)\}_{t=1}^{T}$. \\
\textbf{Output:} - The sequence $\{{\boldsymbol{\beta}}(t)\}_{t=1}^{T}$; \\
\hspace*{3.5em} - The set $\mathcal X_1$; \\
\hspace*{3.5em} - Is Problem \ref{probl1} solvable? {\tt 'Yes'}/{\tt 'No'}. 
\begin{enumerate} 
    \item[\bf 1.] {\em Initialization:} Set:
    \begin{itemize}
        \item $\mybeta(t) = \vect{0}_N, t \in [1,T-1]$, ${\boldsymbol{\beta}}(T) = {\boldsymbol{\alpha}}(T)$;
        \item $\mathcal X_1 = \emptyset$;
        \item $t = T-1$.
    \end{itemize} 
    \item[\bf 2.] {\tt if} $\mybeta(T) = \vect{0}_N$, {\tt then} \\
    \hspace*{1em} Return $(\{{\boldsymbol{\beta}}(t)\}_{t=1}^T, \mathcal X_1, {\tt 'No'})$; \\
    {\tt end if}
    \item[\bf 3.] {\em Iterative procedure:}
    \begin{itemize}
    \item ${\boldsymbol{\beta}}(t) = {\boldsymbol{\alpha}}(t)\odot (L_{tot}^\top {\boldsymbol{\beta}}(t+1))$;
    \item $t \leftarrow t-1$;
    \item {\tt if} $t = 0$, {\tt then} \\
    \hspace*{1em} Go to step {\bf 4}; \\
    {\tt end if}
    \end{itemize}
    \item[\bf 4.] $\mathcal I_1 = {\tt nonzero}({\boldsymbol{\beta}}(1))$; \\
    {\tt for} $j \in \mathcal I_1$, {\tt do} \\
    \hspace*{1em} $\mathcal X_1 \leftarrow \mathcal X_1 \cup \delta_N^j$; \\
    {\tt end for}
    \item[\bf 5.] {\em Conclusions:} \\ {\tt if} $L_{tot}^\top\mybeta(1) > \vect{0}_N$, {\tt then} \\
    \hspace*{1em} Return $(\{{\boldsymbol{\beta}}(t)\}_{t=1}^T, \mathcal X_1,{\tt 'Yes'})$; \\
    {\tt else} \\
    \hspace*{1em} Return $(\{{\boldsymbol{\beta}}(t)\}_{t=1}^T, \mathcal X_1,{\tt 'No'})$; \\
    {\tt end if}
\end{enumerate}
\end{algorithm}

The $j$-th entry of the vector ${\boldsymbol{\beta}}(t)$ is nonzero (i.e., $[{\boldsymbol{\beta}}(t)]_j \ne 0$) if and only if there exists a state sequence $\{\x_k\}_{k=1}^{T}$ 
compatible with the given reference output for which $\x_t = \delta_N^j$.
If we denote by 
 $\mathcal X_t, t \in[1,T]$, the set of the logical vectors $\x_t$  that represent the 
 value at time $t$  of a state trajectory $\{\x_k\}_{k=1}^T$ compatible with the reference output trajectory, i.e., 
\be \label{xt_set}
\mathcal X_t \triangleq \bigl\{ \delta_N^j : [{\boldsymbol{\beta}}(t)]_j \ne 0 \bigr\}, 
\ee
then the set $\mathcal X_1$ can be characterized as 
\be \label{x1_set}
\mathcal X_1 = \bigl\{ \delta_N^j : [{\boldsymbol{\beta}}(1)]_j \ne 0 \bigr\}. 
\ee


So far, we have established only a criterion and an algorithm to determine whether  Problem \ref{probl1} is solvable; however, Theorem \ref{cns_exact} provides no insight into how to construct the solution. To this end, we present Algorithm \ref{alg2}, derived from Algorithm \ref{alg1} (see also Algorithm 1 in \cite{FT_track_BCN}),   which is valid under the hypothesis that Problem \ref{probl1} admits a solution. This modified version accounts for the arbitrariness of the initial condition and stores all possible input values associated with state trajectories that achieve output tracking. 


The algorithm returns the sets $\mathcal T_{xu}(t), t\in [0,T-1]$, each containing all state/input pairs at time $t$  consistent with the reference output  sequence $\{\y_r(t)\}_{t=1}^{T}$, i.e., 
\!\!\begin{subequations} \label{T.set}
\begin{align} 
\!\!\!\mathcal T_{xu}(0) \triangleq & \big\{ (\delta_N^j, \delta_M^i) \in \mathcal L_N \times \mathcal L_M :  \\
&(L\ltimes \delta_M^i \ltimes \delta_N^j)\odot {\boldsymbol{\beta}}(1) \ne \vect{0}_N \big\} \nonumber \\
\!\!\!\mathcal T_{xu}(t) \triangleq & \big\{ (\delta_N^j, \delta_M^i) \in \mathcal L_N \times \mathcal L_M : [\mybeta(t)]_j \ne 0 \ {\rm and} \\
&(L\ltimes \delta_M^i \ltimes \delta_N^j)\odot {\boldsymbol{\beta}}(t+1) \ne \vect{0}_N \big\},  t \in [1, T-1]. \nonumber 
\end{align}
\end{subequations}

\begin{algorithm}[h]
\caption{State/input pairs compatible with the  finite-length  reference output trajectory $\{\y_r(t)\}_{t=1}^{T}$} \label{alg2} 
\smallskip
\textbf{Input:} - A BCN as in \eqref{BCNtot}, i.e., a pair $(L,H)$;\\
\hspace*{2.7em} - \!The sequence $\{\mybeta(t)\}_{t=1}^{T}$  generated by Algorithm \ref{alg1}. \\
\textbf{Output:} The sets $\mathcal T_{xu}(t), \forall t\in [0,T-1]$.  
\begin{enumerate}
    \item[\bf 1.] {\em Initialization:} Set $\mathcal T_{xu}(t) = \emptyset, \forall t \in [0,T-1]$, and \\ $k = 0$. 
    \item[\bf 2.]{\em Iterative procedure:}
    \begin{itemize}
        \item {\tt if} $k = 0$, {\tt then} \\
        \hspace*{1em} $\mathcal I = [1,N]$; \\
        {\tt else} \\
        \hspace*{1em} $\mathcal I = {\tt nonzero}(\mybeta(k))$; \\
        {\tt end if} 
        \item {\tt for} $j\in \mathcal I$, {\tt do} \\
        \hspace*{1em} {\tt for} $i\in[1,M]$, {\tt do} \\
        \hspace*{2.5em} {\tt if} $(L\ltimes \delta_M^i \ltimes \delta_N^j)\odot {\boldsymbol{\beta}}(k+1) \ne \vect{0}_N$, {\tt then} \\
        \hspace*{4em} $\mathcal T_{xu}(k) \leftarrow \mathcal T_{xu}(k) \cup (\delta_N^j, \delta_M^i)$; \\
        \hspace*{2.5em} {\tt end if} \\
        \hspace*{1em} {\tt end for} \\
        {\tt end for}
        \item $k \leftarrow k+1;$
        \item {\tt if} $k = T$, {\tt then} \\
        \hspace*{1em} Go to step {\bf 3}; \\ 
        {\tt else} \\
        \hspace*{1em} Repeat the {\em Iterative procedure}; \\
        {\tt end if}
        \end{itemize}
        \item[\bf 3.] {\em Conclusions:}  Return $\mathcal T_{xu}(t), \forall t \in [0,T-1]$. 
\end{enumerate}
\end{algorithm}

After applying Algorithm \ref{alg2}, to obtain an input sequence that ensures tracking of the reference output trajectory from any initial state, it is sufficient to select, at each time instant $t\in [0,T-1]$, an input value $\ub_t=\delta^i_M$ that forms  with the current state $\x_t=\delta^j_N$ a pair $(\x_t,\ub_t)=(\delta^j_N, \delta^i_M) \in \mathcal T_{xu}(t)$.  
To this end, it is convenient to introduce, for $t\in [0,T-1]$, the following family of sets:
\be
\label{u_track}
\mathcal T_u(t, \x_t) \triangleq \{\ub_t  \in \mathcal L_M : (\x_t,\ub_t) \in \mathcal T_{xu}(t)\}. 
\ee
The following corollary formally describes how the solution can be obtained.

\begin{corollary} \label{input_p1}
Given a BCN \eqref{BCNtot} and a finite-length reference output trajectory $\{\y_r(t)\}_{t=1}^{T}$, with $\y_r(t)\in {\mathcal L}_P$, 
assume that Problem \ref{probl1} is solvable. Then, for every initial state 
$\x(0) = \x_0$, an input sequence $\{\ub(t)\}_{t=0}^{T-1}$ that ensures output tracking is obtained by selecting, at each time instant $t \in [0,T-1]$, the input  $\ub(t) =\ub_t$ 
according to the following recursive procedure. For $t = 0, 1, 
    \dots, T-1$:
\begin{itemize}
    \item 
    $\mathbf{u}_t \in \mathcal{T}_u(t,\x_t)$;
     \item $t \leftarrow t+1$;
    \item $\x_{t} = L \ltimes \ub_{t-1} \ltimes \x_{t-1}$. 
\end{itemize}
\end{corollary}

\begin{remark} \label{no_solv}
Clearly, if Problem \ref{probl1} is not solvable, and thus output tracking of the finite-length trajectory cannot be achieved starting from {\em every} initial state, there may still exist specific initial states for which output tracking is possible. In particular, the set of initial states, say $\mathcal X_0$, that allow output tracking corresponds precisely to those states from which $\mathcal X_1$ is reachable in one step, i.e., 
\be \label{x0_set}
\mathcal X_0 \triangleq \{\delta_N^j \in \mathcal L_N : \exists \ub \in \mathcal L_M \ {\rm s.t.} \ L\ltimes \ub \ltimes \delta_N^j \in \mathcal X_1\}. 
\ee
It is worth noticing that Algorithm \ref{alg2} can also be run in case the finite-length reference output trajectory is compatible with the BCN \eqref{BCNtot}, but ${\mathcal X}_1$ is not reachable in one step from every $\x_0\in {\mathcal L}_N$. When so, one can deduce ${\mathcal X}_0$ as the projection of ${\mathcal T}_{xu}(0)$ on the first variable, namely
$${\mathcal X}_0= \{\x_0\in {\mathcal L}_N : (\x_0, \ub_0) \in {\mathcal T}_{xu}(0), \exists \ub_0\in {\mathcal L}_M\}.$$
\end{remark}

\begin{remark} \label{justify}
It is worth noting that the output tracking problem has been addressed by decomposing it into two subproblems: (1) verifying the compatibility of the reference output trajectory with the given BCN, and (2) ensuring one-step reachability of the initial states associated with compatible state trajectories. This formulation naturally leads to an algorithmic procedure in which the vector sequences $\{\myalpha(t)\}_{t=1}^{T}$ and  $\{\mybeta(t)\}_{t=1}^{T}$ are computed forward and backward in time, respectively, starting from and up to $t = 1$. Alternatively, the same problem could be solved by computing these sequences for $t \in [0,T]$ (as in \cite{FT_track_BCN})  while initializing $\myalpha(0) = \vect{1}_N$, meaning that all possible initial states are considered simultaneously. Proceeding backward in time, under this formulation, would yield a vector $\mybeta(0)$ whose nonzero entries correspond to the initial states from which the reference output trajectory is trackable. Consequently, the problem stated in Problem~\ref{probl1} would be solvable if and only if $\mybeta(0) = \vect{1}_N$. On the other hand, the approach proposed in this paper offers two main advantages. First, by addressing the two subproblems separately, it can be more readily extended to scenarios in which exact tracking is not achievable, yet delayed tracking remains possible - meaning that the set of initial states corresponding to compatible state trajectories can be reached after more than one step (depending on the specific initial state). Second, it provides a framework that is naturally adaptable to the case of periodic reference output trajectories, addressed in Section \ref{sec:periodic}. 
\end{remark}

\section{Output Tracking of Periodic Trajectories} \label{sec:periodic}

Tracking periodic trajectories is a fundamental problem in control theory, as many practical applications require precise repetition of specific motions or signals. Therefore, in this section, we build upon the results obtained in the previous section concerning tracking of finite-length trajectories, and extend them to the more interesting case of periodic trajectory tracking. 

Consider a periodic output sequence of minimal period $T\in \mathbb Z_+, T \ge 1,$ given, for every $t\in {\mathbb Z}_+$, by 
\be \label{ref_per}
\y_r(t) = \begin{cases}
 \delta_P^{i_1}, \quad t = 1+kT, \\
\delta_P^{i_2}, \quad t = 2+kT, \\
\vdots \\
 \delta_P^{i_T}, \quad t = (k+1)T,
\end{cases}
\ee
with $k \in \mathbb Z_+$.  Similarly to Definition \ref{trackability_def}, a  periodic output trajectory \eqref{ref_per} is trackable from a given initial condition $\x(0) = \x_0 \in {\mathcal L}_N$ if there exists a control sequence $\{\ub(t)\}_{t\in \mathbb Z_+}$
such that $\y(t,\x_0,\{\ub(\tau)\}_{\tau=0}^{t-1}) = \y_r(t), \forall t \in \mathbb Z_+, t\ge 1$.
We aim to solve the following problem. 

\begin{problem} \label{probl3}
Given a periodic reference output trajectory \eqref{ref_per}, under what conditions is it trackable starting from any initial condition?
\end{problem}

As in Section \ref{sec:track_mb}, we begin by establishing  necessary and sufficient conditions for the problem solvability, and then we present an algorithm to compute the solution, if one exists.

It turns out that Problem \ref{probl3} is solvable if and only if Problem \ref{probl1} is solvable for the finite-length reference output trajectory given by $\{\y_r(t)\}_{t=1}^{T} = \{\delta_P^{i_t}\}_{t=1}^{T}$, that is, the restriction of the periodic trajectory to a single period. Note that the equivalence strongly depends on the fact that in both problems output tracking is required starting from {\em every} initial state. The proof relies on the sets ${\mathcal X}_1$ and ${\mathcal X}_T$, defined in the previous section and obtained through Algorithm \ref{alg1}.

\begin{proposition} \label{p1_p2}
Given a BCN \eqref{BCNtot} and a periodic reference output trajectory of minimal period $T$ as in \eqref{ref_per}, the following facts are equivalent. 
\begin{itemize}
\item[i)] Problem \ref{probl3} is solvable, i.e., the periodic  output trajectory in \eqref{ref_per} is trackable   from every $\x(0) = \x_0 \in \mathcal L_N$. 
\item[ii)]  Problem \ref{probl1} is solvable for the reference output trajectory $\{\y_r(t)\}_{t=1}^{T}$, i.e., the finite-length trajectory $\{\y_r(t)\}_{t=1}^{T}$ is trackable   from every  $\x(0)=\x_0 \in \mathcal L_N$.
\end{itemize}
\end{proposition}

\begin{proof}
The implication {\em i) $\Rightarrow$ ii)} is trivial. 

Conversely, the implication {\em ii) $\Rightarrow$ i)} is a direct consequence of Theorem \ref{cns_exact}. Indeed, according to Theorem \ref{cns_exact}, if the finite-length trajectory $\{\y_r(t)\}_{t=1}^{T}$ is trackable starting from every $\x(0) = \x_0 \in \mathcal L_N$, then the following conditions hold: 
\begin{enumerate}
\item[1)] The set $\mathcal X_1$ is non-empty, and hence the finite output sequence $\{\y_r(t)\}_{t=1}^{T}$ can be generated by at least one state trajectory, say $\{\x_t\}_{t=1}^T$. 
\item[2)] The set $\mathcal X_1$ is reachable in one step from every state in $\mathcal L_N$, and in particular, from all states in $\mathcal X_T$. 
\end{enumerate}
Together, these two conditions ensure the following: first, for every $\x(0) = \x_0\in {\mathcal L}_N$ there exists an input $\ub(0)=\ub_0 \in {\mathcal L}_M$ such that $\x_1 = L\ltimes \ub_0\ltimes \x_0 \in {\mathcal X}_1$; next, there exists a state trajectory from  $\x_1$ to some $\x_T\in {\mathcal X}_T$ that generates the desired reference output over the interval $[1,T]$; finally, condition 2) guarantees that there exists a transition from   $\x_T$ back to one of the initial states in ${\mathcal X}_1$, thereby enabling the recursive generation of the same reference output in each subsequent interval of length $T$. 
\end{proof}

 
We now focus on the situation in which the given periodic output trajectory is {\em not} trackable starting from {\em every} initial state. 
By Proposition \ref{p1_p2}, this directly implies that also the finite-length trajectory obtained by restricting the periodic trajectory to a single period can not be tracked from every initial state. 
However, we assume that there exists at least one state sequence $\{\x_t\}_{t=1}^T$ such that $H\x_t = \y_r(t), \forall t \in [1,T]$, and therefore the finite-length output trajectory $\{\y_r(t)\}_{t=1}^T$ can be generated by the BCN \eqref{BCNtot}.

The existence of a state trajectory producing the desired output on a single period does not ensure that the finite-length output sequence can be repeated indefinitely.
For this to happen, 
it is necessary and sufficient that there exists a family of (finite-length) state trajectories of the BCN, say 
\be \label{fam}
{\mathcal F}_T \triangleq \{ \{\x_t^{(\ell)}\}_{t=1}^T, \ell\in [1, \bar \ell]\},
\ee
satisfying the following two conditions:
\begin{itemize}
\item[(C1)] for each $\ell\in [1,\bar \ell]$ and each $t\in [1,T]$, $H\x_t^{(\ell)} = \y_r(t)$, namely each  state trajectory in ${\mathcal F}_T$ is compatible with  the finite-length  output trajectory $\{\y_r(t)\}_{t=1}^T$;
\item[(C2)] for each $\ell\in [1,\bar \ell]$ the state $\x_T^{(\ell)}$ has a successor in the set $\{\x_1^{(\ell)}, \ell\in [1,\bar \ell]\}.$
\end{itemize}
Using the analysis carried out in the previous section and   definition \eqref{xt_set} of the sets ${\mathcal X}_t, t\in [1,T],$ conditions (C1) and (C2) are verified if not only ${\mathcal X}_T$ is non-empty, but also 
all the successors of the states in ${\mathcal X}_T$ belong to ${\mathcal X}_1.$
It is worth noticing, however, that the fact that this does not happen after running Algorithm \ref{alg1} does not mean that
the aforementioned family of state trajectories
${\mathcal F}_T$  (satisfying (C1) and (C2)) does not exist. Indeed, the algorithm needs to be modified  in such a way that, as far as at least a subset of the states in ${\mathcal X}_T$ has a successor in ${\mathcal X}_1$, one can remove
all the state trajectories  
that end in the states of ${\mathcal X}_T$ with no successor in ${\mathcal X}_1$ and perform the same check again.\footnote{It is important to underline that removing a state sequence does not necessarily mean removing all the states that are part of such trajectory. In fact, if such states are part of other state trajectories compatible with $\{\y_r(t)\}_{t=1}^T$ and ending in a state of ${\mathcal X}_T$ with a successor in ${\mathcal X}_1$ they will be preserved.} This may require several rounds, but since at every round the number of considered final states decreases, the procedure necessarily stops either because condition (C2) holds or because the set of state trajectories has become empty.

Algorithm \ref{alg3}, below, receives as input the sequence $\{ \myalpha(t)\}_{t=1}^{T+1}$, generated according to  \eqref{a_seq} but extended till time $T+1$ rather than $T$. Then, based on $\{ \myalpha(t)\}_{t=1}^{T+1}$, it generates the vector sequence $\{\mybeta_1(t)\}_{t = 1}^{T+1}$  and  the associated set $\mathcal X_1^{(1)}$,  following the same procedure as in Algorithm \ref{alg1}, as well as the set ${\mathcal X}^{(1)}_{T+1}.$ Three cases may arise:\\
$\bullet$  If ${\mathcal X}^{(1)}_{1} = \emptyset$ (equivalently, $\mybeta_1(t) = \vect{0}_N, \forall t \in [1,T+1]$) then the algorithm stops since the  periodic trajectory cannot be generated by the BCN. \\
$\bullet$  If\footnote{Note that if ${\mathcal X}^{(1)}_{T+1}=\emptyset$, then necessarily ${\mathcal X}^{(1)}_{1} = \emptyset$.} ${\mathcal X}^{(1)}_{1} \ne \emptyset$ and 
${\mathcal X}^{(1)}_{T+1} \subseteq {\mathcal X}^{(1)}_{1}$,  the periodic output trajectory is compatible with the BCN and the algorithm stops. \\
$\bullet$ Finally, if ${\mathcal X}^{(1)}_{1} \ne \emptyset$ but it does not include
${\mathcal X}^{(1)}_{T+1}$, then 
the states of ${\mathcal X}^{(1)}_{T+1}$ that do not belong to ${\mathcal X}^{(1)}_{1}$ are removed and a new iteration starts. Specifically, a new vector $\mybeta_2(T+1)$ is generated by zeroing in $\mybeta_1(T+1)$ the entries that are zero in $\mybeta_1(1).$
At this stage, the vector sequence $\{\mybeta_2(t)\}_{t=1}^{T+1}$ is generated backward from $\mybeta_2(T+1)$: for every $t = T, T-1, \dots, 1$ the vector $\mybeta_2(t)$ is obtained from $\mybeta_1(t)$ by replacing with  zeros the entries associated with those states that have no longer 
 successors corresponding to the nonzero entries of $\mybeta_2(t+1)$.
Again, the sets ${\mathcal X}^{(2)}_{T+1}$ and ${\mathcal X}^{(2)}_{1}$ are compared and  three cases - as before - can arise.  If ${\mathcal X}^{(2)}_{T+1} \ne \emptyset$ but it is not included in 
${\mathcal X}^{(2)}_{1}$, then the algorithm starts a new round, generating the vectors $\{\mybeta_3(t)\}_{t=1}^{T+1}$ and the sets ${\mathcal X}^{(3)}_{1}$ and ${\mathcal X}^{(3)}_{T+1}$.
This procedure will come to an end, after a finite number of steps  $k^*$, with either ${\mathcal X}^{(k^*)}_{1} = \emptyset$ or 
${\mathcal X}^{(k^*)}_{T+1} \subseteq {\mathcal X}^{(k^*)}_{1}$, thus 
establishing whether 
the periodic output trajectory is compatible with the BCN.



\begin{algorithm}[h]
\caption{Compatibility of the periodic reference output trajectory in \eqref{ref_per}} \label{alg3} 
\smallskip
\textbf{Input:}  - A BCN as in \eqref{BCNtot}, i.e., a pair $(L,H)$;\\
\hspace*{2.8em} - The sequence $\{\myalpha(t)\}_{t=1}^{T+1}$. \\
\textbf{Output:} - The sequence $\{\mybeta_{k^*}(t)\}_{t=1}^{T+1}$; \\
\hspace*{3.5em} - The set $\mathcal X_1^{(k^*)}$; \\
\hspace*{3.5em} - Is the periodic reference output trajectory \eqref{ref_per} \\ 
\hspace*{4.2em} compatible  with the BCN \eqref{BCNtot}? {\tt 'Yes'}/{\tt 'No'}.
\begin{enumerate}
    \item[\bf 1.] {\em Initialization:} 
    \begin{itemize}
        \item Set $k = 1$.
        \item Apply Algorithm \ref{alg1} to the sequence $\{\myalpha(t)\}_{t=1}^{T+1}$ to obtain the sequence $\{\mybeta_1(t)\}_{t=1}^{T+1}$. 
        \item Determine the sets  $\mathcal X_1^{(1)}$ and $\mathcal X_{T+1}^{(1)}$ as in \eqref{xt_set} using the sequence $\{\mybeta_1(t)\}_{t=1}^{T+1}$.
    \end{itemize} 
    \item[\bf 2.] {\em Iterative procedure:}
    \begin{itemize}
        \item {\tt if} $\mathcal X_{T+1}^{(k)} \subseteq \mathcal X_1^{(k)}$ {\tt or} $\mathcal X_1^{(k)} = \emptyset$, {\tt then} \\ 
        \hspace*{1em} Set $k^* = k$ and go to step {\bf 3}; \\
        {\tt end if}
        \item $k \leftarrow k+1$;
        \item $\mybeta_k(T+1) = \mybeta_{k-1}(T+1) \odot \mybeta_{k-1}(1)$;
        \item $t = T$;
        \item {\tt while} $t > 0$, {\tt do} \\\hspace*{1em}$\mybeta_k(t) = \mybeta_{k-1}(t)\odot (L_{tot}^\top \mybeta_k(t+1))$; \\
        \hspace*{1em} $t \leftarrow t-1$;
        \item Determine the sets  $\mathcal X_1^{(k)}$ and $\mathcal X_{T+1}^{(k)}$ as in \eqref{xt_set} using the sequence $\{\mybeta_k(t)\}_{t=1}^{T+1}$;
        \item Repeat the {\em Iterative procedure}.
    \end{itemize}
    \item[\bf 3.] {\em Conclusions:} \\
    {\tt if} $\mathcal X_1^{(k)} = \emptyset$, {\tt then} \\
    \hspace*{2em} Return $(\mathcal X_1^{(k^*)}, \{\mybeta_{k^*}(t)\}_{t=1}^{T+1}, {\tt 'No'})$; \\
    {\tt else} \\
    \hspace*{2em} Return $(\mathcal X_1^{(k^*)}, \{\mybeta_{k^*}(t)\}_{t=1}^{T+1}, {\tt 'Yes'})$; \\
    {\tt end if}
\end{enumerate}
\end{algorithm}

\begin{remark} \label{k_star}
As already mentioned, Algorithm \ref{alg3} terminates in a finite number of steps $k^*$. 
The value of $k^*$ is upper bounded by the quantity $|{\mathcal C}_{i_1}| + 1$. Indeed, the total number of iterations depends on the number of nonzero entries in $\mybeta_1(T+1)$. Since the algorithm is initialized with $\mybeta_1(T+1) = \myalpha(T+1)$, the number of nonzero entries in $\mybeta_1(T+1)$ cannot exceed the number of those in ${\bf v}(T+1) = {\bf v}(1)$, which is equal to the cardinality of $\mathcal C_{i_1}$. Consequently, the maximum number of iterations occurs when $\mybeta_1(T+1)$ contains $|{\mathcal C}_{i_1}|$ nonzero entries, at each iteration $k$ of the algorithm
a single entry of $\mybeta_{k-1}(T+1)$ is replaced with zero in $\mybeta_k(T+1)$, and the procedure terminates with $\mybeta_{k^*}(1) = \vect{0}_N$, yielding $k^* = |{\mathcal C}_{i_1}| + 1$. It is important to observe, however, that when Algorithm \ref{alg3} terminates with an affirmative answer, the total number of iterations cannot exceed $|{\mathcal C}_{i_1}|$. 
\end{remark}

One of the outcomes of Algorithm \ref{alg3} is the set $\mathcal X_1^{(k^*)}$, which contains the initial states $\x_1$ of the state trajectories $\{\x_t\}_{t = 1}^{+\infty}$ compatible with the given periodic reference output trajectory \eqref{ref_per}. To solve the periodic output tracking problem for a {\em  specific} initial condition $\x(0) = \x_0$, we also need to ensure that such set is reachable in one step from $\x_0$. We denote by $\mathcal X_0^p$ the set of initial states starting from which the given periodic output trajectory is trackable, i.e., 
\be \label{xop}
\mathcal X_0^p \triangleq \{\delta_N^j \in \mathcal L_N : \exists \ub \in \mathcal L_M \ {\rm s.t.} \ L \ltimes \ub \ltimes \delta_N^j \in \mathcal X_1^{(k^*)}\}.
\ee
Clearly, the set $\mathcal X_0^p$ in \eqref{xop} is a subset of the set $\mathcal X_0$, defined in \eqref{x0_set} and computed for the finite-length trajectory $\{\y_r(t)\}_{t=1}^{T}$.  
This trivially follows from the fact that if the periodic reference output trajectory can be tracked from some initial condition $\x_0$, then the finite-length trajectory corresponding to a single period of that periodic trajectory can also be tracked starting from $\x_0$.

In the 
following theorem, we provide necessary and sufficient conditions for the solvability of the output tracking problem of a periodic trajectory starting from a specific initial condition $\x(0) = \x_0$.  The proof follows immediately from the previous analysis and hence is omitted.

\begin{theorem} \label{cns_periodic}
Given a BCN \eqref{BCNtot}, a periodic reference output trajectory of minimal period $T$ as in \eqref{ref_per}, and an initial condition $\x(0) = \x_0$, the periodic reference output trajectory is trackable starting from $\x_0$  if and only if the following conditions hold:
\begin{itemize} 
\item[i)] The set $\mathcal X_1^{(k^*)}$ is non-empty and $\mathcal X_{T+1}^{(k^*)} \subseteq \mathcal X_1^{(k^*)}$.  
\item[ii)] There exists $\ub \in \mathcal L_M$ such that $L \ltimes \ub \ltimes \x_0 \in \mathcal X_1^{(k^*)}$. 
\end{itemize}
\end{theorem}



The next step is to identify all state/input trajectories that achieve output tracking of the given periodic trajectory for the admissible initial conditions, that is, for the states in $\mathcal X_0^p$. 
To this end, the algorithmic procedure described in Section \ref{sec:track_mb} can be suitably adapted to the case of periodic trajectory tracking. Specifically, we can apply Algorithm \ref{alg2} (see Remark \ref{no_solv}) to the sequence $\{\mybeta_{k^*}(t)\}_{t=1}^{T+1}$ returned by Algorithm \ref{alg3}.
 The algorithm returns the sets $\mathcal T_{xu}(t), t \in [0,T],$ and we can define the sets  ${\mathcal T}_u(t, \x_t),  t \in [0,T],$ according to \eqref{u_track}. \\
 It is worth underlying that
 $\mathcal T_{xu}(T) \subseteq \mathcal T_{xu}(0)$. Indeed, 
 for every $\x \in {\mathcal X}_T^{(k^*)}$, ${\mathcal T}_u(T,\x) = {\mathcal T}_u(0,\x).$  However, for the states $\x\in {\mathcal X}_0^p \setminus {\mathcal X}_T^{(k^*)}$ the set
 ${\mathcal T}_u(T,\x)$  is empty, while ${\mathcal T}_u(0,\x)$ is not. So, in the sequel we can drop ${\mathcal T}_{xu}(T)$ and ${\mathcal T}_u(T, \x)$ and only use the analogous sets for $t=0$.\\
 We are now in a position to state the following corollary whose proof easily follows  from the previous discussion. 

\begin{corollary} \label{u_per}
Given a BCN \eqref{BCNtot}, a periodic reference output trajectory of minimal period period $T$ as in \eqref{ref_per}, and an initial condition $\x(0) =\x_0\in {\mathcal L}_N$, assume that the periodic reference output trajectory \eqref{ref_per} is trackable from $\x_0$.  
Let $\mathcal T_{xu}(t), t \in [0,T-1]$, be the sets of state/input values generated by Algorithm \ref{alg2}, and 
compute from them the sets $\mathcal T_u(\cdot, \cdot)$ using \eqref{u_track}.  
\\
Then, an input sequence $\{\ub(t)\}_{t\in \mathbb Z_+}$ that ensures  tracking of the periodic trajectory \eqref{ref_per} is obtained  by selecting, at each time instant $t \in \mathbb Z_+$, the input  $\ub(t) =\ub_t$ according to the following recursive procedure. For $t = 0, 1, 2, 
    \dots$:
\begin{itemize}
    \item 
    $\mathbf{u}_t \in \mathcal{T}_u(t\bmod T,\x_t)$;
     \item $t \leftarrow t+1$;
    \item $\x_{t} = L \ltimes \ub_{t-1} \ltimes \x_{t-1}$. 
\end{itemize}
\end{corollary}
\smallskip 

We now demonstrate the validity of the proposed algorithmic procedure through a small numerical example, in which all the steps required to compute the solution can be explicitly written. 

\begin{example} \label{ex}
Consider a BCN \eqref{BCNtot} with $N = 6$ , $M = 2$ and $P = 2$, described by the following matrices: 
\begin{eqnarray*}
L &=& \left[\begin{array}{cccccc|cccccc}
\delta_6^2 \!&\! \delta_6^2 \!&\! \delta_6^4 \!&\! \delta_6^5 \!&\! \delta_6^5 \!&\! \delta_6^5 & \delta_6^4 \!&\! \delta_6^3 \!&\!
\delta_6^1 \!&\! \delta_6^2 \!&\! 
\delta_6^6 \!&\! \delta_6^5
\end{array}\right], \\
H &=& \left[\begin{array}{cccccc}
\delta_2^1 & \delta_2^1 & \delta_2^2 & \delta_2^1 & \delta_2^2 & \delta_2^1
\end{array}\right]. 
\end{eqnarray*}
Assume that the periodic reference output trajectory to be tracked is given by 
$$
\y_r(t) = \begin{cases}
 \delta_2^{1}, \quad t = 1+3k, \\
 \delta_2^{1}, \quad t = 2+3k, \\
\delta_2^{2}, \quad t = 3(k+1), \\
\end{cases} \quad k \in \mathbb Z_+
$$
whose minimal period is $T = 3$.

Before applying the algorithmic procedure to determine whether the given periodic trajectory is trackable - and, if so, to compute the corresponding solution - we first identify the indistinguishability classes of the periodic sequence together with their vectorized representations $\{{\bf v}(t)\}_{t=1}^3$, obtaining: ${\bf v}(1) = {\bf v}(2) = \delta_6^1+\delta_6^2+\delta_6^4+\delta_6^6$, ${\bf v}(3) = \delta_6^3+\delta_6^5$. 

It is not difficult to verify that the conditions in Theorem \ref{cns_exact} are not satisfied.
According to Proposition \ref{p1_p2}, this implies that the periodic reference output trajectory cannot be tracked from every initial condition. Therefore, after computing the sequence $\{\myalpha(t)\}_{t=1}^4$, 
Algorithm \ref{alg3} must be executed to determine the set of initial states (if any) from which the given periodic output trajectory is trackable. The initialization of the algorithm (i.e., $k = 1$) returns the following sequence: 
\begin{eqnarray*}
\mybeta_1(1) \!\!\!\!&=&\!\!\!\! \begin{bmatrix}
    * \!&\! * \!&\! 0 \!&\! * \!&\! 0 \!&\! 0
    \end{bmatrix}^\top, \ \mybeta_1(2) =  \begin{bmatrix}
    0 \!&\! * \!&\! 0 \!&\! * \!&\! 0 \!&\!  0
    \end{bmatrix}^\top, \\ \mybeta_1(3) \!\!\!\!&=&\!\!\!\!  \begin{bmatrix}
    0 \!&\! 0 \!&\! * \!&\! 0 \!&\! * \!&\! 0
    \end{bmatrix}^\top, \ \mybeta_1(4) =  \begin{bmatrix}
    * \!&\! 0 \!&\! 0 \!&\! * \!&\! 0 \!&\! *
    \end{bmatrix}^\top,
\end{eqnarray*}
where $*$ denotes a nonzero entry, as the precise value is not relevant. Since $\mathcal X_1^{(1)} \ne \emptyset$, but $\mathcal X_4^{(1)} \not\subseteq \mathcal X_1^{(1)}$, we 
repeat the iterative procedure 
(i.e., $k = 2$), obtaining:
\begin{eqnarray*}
\mybeta_2(1) \!\!\!\!&=&\!\!\!\! \begin{bmatrix}
    * \!&\! * \!&\! 0 \!&\! * \!&\! 0 \!&\! 0
    \end{bmatrix}^\top, \ \mybeta_2(2) =  \begin{bmatrix}
    0 \!&\! * \!&\! 0 \!&\! 0 \!&\! 0 \!&\!  0
    \end{bmatrix}^\top, \\ \mybeta_2(3) \!\!\!\!&=&\!\!\!\!  \begin{bmatrix}
    0 \!&\! 0 \!&\! * \!&\! 0 \!&\! 0 \!&\! 0
    \end{bmatrix}^\top, \ \mybeta_2(4) =  \begin{bmatrix}
    * \!&\! 0 \!&\! 0 \!&\! * \!&\! 0 \!&\! 0
    \end{bmatrix}^\top.
\end{eqnarray*}
Now, $\mathcal X_4^{(2)} \subseteq \mathcal X_1^{(2)}$, and hence we set $k^* = 2$. \\
The set of initial states for which $\mathcal X_1^{(2)} = \{\delta_6^1,\delta_6^2,\delta_6^4\}$ is reachable in one step (and hence the periodic trajectory is trackable) is given by 
$
\mathcal X_0^p = \{\delta_6^1, \delta_6^2, \delta_6^3, \delta_6^4\}. 
$
By applying Algorithm \ref{alg2} to the sequence $\{\mybeta_{2}(t)\}_{t=1}^{T+1}$ returned by Algorithm \ref{alg3}, we get 
\begin{eqnarray*}
\mathcal T_{xu}(0) \!\!\!&=&\!\!\! \{(\delta_6^1,\delta_2^1); (\delta_6^1,\delta_2^2); (\delta_6^2,\delta_2^1);  (\delta_6^3,\delta_2^1); \\
\!\!\!&&\!\!\!(\delta_6^3,\delta_2^2);(\delta_6^4,\delta_2^2)\}, \\
\mathcal T_{xu}(1) \!\!\!&=&\!\!\! \{(\delta_6^1,\delta_2^1); (\delta_6^2,\delta_2^1);
(\delta_6^4,\delta_2^2)\}, \\
\mathcal T_{xu}(2) \!\!\!&=&\!\!\! \{ (\delta_6^2,\delta_2^2)\}. 
\end{eqnarray*}
Therefore, if for every $\x_0 \in \mathcal X_0^p$, the input sequence $\{\ub(t)\}_{t\in \mathbb Z_+}$ is selected according to Corollary \ref{u_per}, the tracking objective is achieved.  

\end{example}

\section{Conclusions} \label{concl}

In this paper, we addressed the output tracking problem for BCNs, starting with the case of finite-length trajectories and then extending the analysis to the more challenging and practically relevant problem of periodic trajectory tracking. Building upon the results in \cite{FT_track_BCN}, we revisited and generalized the algorithmic procedure to handle a wider range of tracking scenarios.
We provided necessary and sufficient conditions for the solvability of both problems, requiring trackability from any initial state, and developed algorithms for computing all corresponding control sequences. Furthermore, we investigated the case, where the tracking solution does not exist for all initial conditions, which led to a modified algorithmic approach specifically tailored for the periodic case. Finally, a small numerical example was presented to demonstrate the effectiveness of the proposed procedure.


\begin{thebibliography}{100}

\bibitem{STP2001}
D.~Cheng.
\newblock Semi-tensor product of matrices and its some applications to
  {M}orgen's problem.
\newblock {\em Science in China Series: Information Sciences}, 44(3):195--212,
  2001.

\bibitem{Cheng2009}
D.~Cheng.
\newblock Input-state approach to {B}oolean {N}etworks.
\newblock {\em IEEE Trans. Neural Networks}, 20, (3):512 -- 521, 2009.

\bibitem{ChengContr}
D.~Cheng and H.~Qi.
\newblock Controllability and observability of {B}oolean control networks.
\newblock {\em Automatica}, 45(7):1659--1667, 2009.

\bibitem{lin-rep-dyn-Boolean}
D.~Cheng and H.~Qi.
\newblock A linear representation of dynamics of {B}oolean networks.
\newblock {\em IEEE Trans. Automatic Control}, 55:2251--2258, 2010.

\bibitem{BCNCheng}
D.~Cheng, H.~Qi, and Z.~Li.
\newblock {\em Analysis and control of {B}oolean networks}.
\newblock Springer-Verlag, London, 2011.

\bibitem{EF_MEV_BCN_obs2012}
E.~Fornasini and M.~E. Valcher.
\newblock Observability, reconstructibility and state observers of {B}oolean
  control networks.
\newblock {\em IEEE Tran. Aut. Contr.}, 58 (6):1390 -- 1401, 2013.

\bibitem{EF_MEV_BCN_Aut}
E.~Fornasini and M.~E. Valcher.
\newblock On the periodic trajectories of {B}oolean {C}ontrol {N}etworks.
\newblock {\em Automatica}, 49:1506--1509, 2013.

\bibitem{ACC2014_BCN}
E.~Fornasini and M.E. Valcher.
\newblock Feedback stabilization, regulation and optimal control of {B}oolean
  control networks.
\newblock In {\em Proc. of the 2014 American Control Conference}, pages
  1993--1998, Portland, OR, 2014.

\bibitem{FaultDetTAC15}
E.~Fornasini and M.E. Valcher.
\newblock Fault detection analysis of {B}oolean control networks.
\newblock {\em IEEE Trans. Automatic Control}, 60 (10):2734--2739, 2015.

\bibitem{EFMEV_JCD}
E.~Fornasini and M.E. Valcher.
\newblock Recent developments in {B}oolean networks control.
\newblock {\em Journal of Control and Decision}, 3:1--18, 2016.

\bibitem{LASCHOV2012}
D.~Laschov and M.~Margaliot.
\newblock Controllability of {B}oolean control networks via the
  {P}erron–{F}robenius theory.
\newblock {\em Automatica}, 48(6):1218--1223, 2012.

\bibitem{Li2015}
H.~Li, Y.~Wang, and L.~Xie.
\newblock Output tracking control of {B}oolean control networks via state
  feedback: Constant reference signal case.
\newblock {\em Automatica}, 59:54--59, 2015.

\bibitem{Li17}
H.~Li, L.~Xie, and Y.~Wang.
\newblock Output regulation of {B}oolean control networks.
\newblock {\em IEEE Transactions on Automatic Control}, 62(6):2993--2998, 2017.

\bibitem{Liu23}
Z.~Liu, Y.~Wang, Y.~Liu, and Q.~Ruan.
\newblock Reference trajectory output tracking for {B}oolean control networks
  with controls in output.
\newblock {\em Mathematical Modelling and Control}, 3(3):256--266, 2023.

\bibitem{PanEtAl2023}
Y.~Pan, S.~Fu, J.~Wang, and W.~Zhang.
\newblock Optimal output tracking of {B}oolean control networks.
\newblock {\em Information Sciences}, 626:524--536, 2023.

\bibitem{track_Cheng}
X.~Zhang, Y.~Wang, and D.~Cheng.
\newblock Output tracking of {B}oolean control networks.
\newblock {\em IEEE Transactions on Automatic Control}, 65(6):2730--2735, 2020.

\bibitem{FT_track_BCN}
Z.~Zhang, T.~Leifeld, and P.~Zhang.
\newblock Finite horizon tracking control of {B}oolean control networks.
\newblock {\em IEEE Transactions on Automatic Control}, 63(6):1798--1805, 2018.

\end{thebibliography}

\end{document}